\documentclass{article}

\title{Dimension Spectra of Lines}
\author{
	Neil Lutz\footnote{Research supported in part by National Science Foundation Grant 1445755.}\\
	Department of Computer Science, Rutgers University\\
	Piscataway, NJ 08854, USA\\
	\texttt{njlutz@rutgers.edu}
	\and
	D. M. Stull\footnote{Research supported in part by National Science Foundation Grants 1247051 and 1545028.}\\
	Department of Computer Science, Iowa State University\\
	Ames, IA 50011, USA\\
	\texttt{dstull@iastate.edu}
}

\usepackage{amsmath}
\usepackage{amsthm}
\usepackage{amssymb}
\usepackage{mdframed}
\usepackage{mathtools}
\usepackage{enumitem}
\usepackage{placeins}
\usepackage{url}
\usepackage{comment}
\usepackage{tikz}

\newtheorem{thm}{Theorem}
\newtheorem{obs}[thm]{Observation}
\newtheorem{lem}[thm]{Lemma}

\newtheorem{cor}[thm]{Corollary}

\DeclareMathOperator{\Dim}{Dim}

\DeclareMathOperator{\spec}{sp}

\newcommand{\R}{\mathbb{R}}

\newcommand{\N}{\mathbb{N}}
\newcommand{\Q}{\mathbb{Q}}

\newcommand{\ve}{\varepsilon}

\begin{document}
	\maketitle
	
	\begin{abstract}
		This paper investigates the algorithmic dimension spectra of lines in the Euclidean plane. Given any line $L$ with slope $a$ and vertical intercept $b$, the dimension spectrum $\spec(L)$ is the set of all effective Hausdorff dimensions of individual points on $L$.  We draw on Kolmogorov complexity and geometrical arguments to show that if the effective Hausdorff dimension $\dim(a, b)$ is equal to the effective packing dimension $\Dim(a, b)$, then $\spec(L)$ contains a unit interval. We also show that, if the dimension $\dim(a, b)$ is at least one, then $\spec(L)$ is infinite. Together with previous work, this implies that the dimension spectrum of any line is infinite.
	\end{abstract}
	
	\section{Introduction}
	
	Algorithmic dimensions refine notions of algorithmic randomness to quantify the density of algorithmic information of individual points in continuous spaces. The most well-studied algorithmic dimensions for a point $x\in\R^n$ are the \emph{effective Hausdorff dimension}, $\dim(x)$, and its dual, the \emph{effective packing dimension}, $\Dim(x)$~\cite{Lutz03b,AHLM07}. These dimensions are both algorithmically and geometrically meaningful~\cite{DowHir10}. In particular, the quantities $\sup_{x\in E}\dim(x)$ and $\sup_{x\in E}\Dim(x)$ are closely related to classical Hausdorff and packing dimensions of a set $E\subseteq\R^n$~\cite{Hitc05,LutLut17}, and this relationship has been used to prove nontrivial results in classical fractal geometry using algorithmic information theory~\cite{Reim08,LutLut17,LutStu17}.
	
	Given the pointwise nature of effective Hausdorff dimension, it is natural to investigate not only the supremum $\sup_{x\in E}\dim(x)$ but the entire \emph{(effective Hausdorff) dimension spectrum} of a set $E \subseteq \R^n$, i.e., the set
	\[\spec(E)=\{\dim(x):x\in E\}\,.\]
	The dimension spectra of several classes of sets have been previously investigated. Gu, et al. studied the dimension spectra of randomly selected subfractals of self-similar fractals~\cite{GLMM14}. Dougherty, et al. focused on the dimension spectra of random translations of Cantor sets~\cite{DLMT14}. In the context of symbolic dynamics, Westrick has studied the dimension spectra of subshifts~\cite{Westrick14}.
	
	This work concerns the dimension spectra of lines in the Euclidean plane $\R^2$. Given a line $L_{a,b}$ with slope $a$ and vertical intercept $b$, we ask what $\spec(L_{a,b})$ might be. It was shown by Turetsky that, for every $n\geq 2$, the set of all points in $\R^n$ with effective Hausdorff 1 is connected, guaranteeing that $1\in\spec(L_{a,b})$.
	In recent work~\cite{LutStu17}, we showed that the dimension spectrum of a line in $\R^2$ cannot be a singleton. By proving a general lower bound on $\dim(x,ax+b)$, which is presented as Theorem~\ref{thm:boundingMainThm} here, we demonstrated that
	\[\min\{1,\dim(a,b)\}+1\in\spec(L_{a,b})\,.\]
	Together with the fact that $\dim(a,b)=\dim(a,a^2+b)\in\spec(L_{a,b})$ and Turetsky's result, this implies that the dimension spectrum of $L_{a,b}$ contains both endpoints of the unit interval $[\min\{1,\dim(a,b)\},\min\{1,\dim(a,b)\}+1]$.
	
	Here we build on that work with two main theorems on the dimension spectrum of a line. Our first theorem gives conditions under which the entire unit interval must be contained in the spectrum. We refine the techniques of~\cite{LutStu17} to show in our main theorem (Theorem~\ref{thm:spectraMainTheorem}) that, whenever $\dim(a,b)=\Dim(a,b)$, we have
	\[[\min\{1,\dim(a,b)\},\min\{1,\dim(a,b)\}+1]\subseteq\spec(L_{a,b})\,.\]
	Given any value $s\in[0,1]$, we construct, by padding a random binary sequence, a value $x\in\R$ such that $\dim(x, ax + b) = s + \min\{\dim(a, b), 1\}$. Our second main theorem shows that the dimension spectrum $\spec(L_{a,b})$ is infinite for every line such that $\dim(a, b)$ is at least one. Together with Theorem~\ref{thm:boundingMainThm}, this shows that the dimension spectrum of \textit{any} line has infinite cardinality.
	
	We begin by reviewing definitions and properties of algorithmic information in Euclidean spaces in Section~\ref{sec:prelim}. In Section~\ref{sec:approach}, we sketch our technical approach and state our main technical lemmas; their proofs are deferred to the appendix. In Section~\ref{ssec:mainthm} we prove our first main theorem and state our second main theorem, whose proof is deferred to the appendix. We conclude in Section~\ref{sec:conc} with a brief discussion of future directions.
	
	\section{Preliminaries}\label{sec:prelim}
	
	\subsection{Kolmogorov Complexity in Discrete Domains}
	The \emph{conditional Kolmogorov complexity} of binary string $\sigma\in\{0,1\}^*$ given a binary string $\tau\in\{0,1\}^*$ is the length of the shortest program $\pi$ that will output $\sigma$ given $\tau$ as input. Formally, it is
	\[K(\sigma|\tau)=\min_{\pi\in\{0,1\}^*}\left\{\ell(\pi):U(\pi,\tau)=\sigma\right\}\,,\]
	where $U$ is a fixed universal prefix-free Turing machine and $\ell(\pi)$ is the length of $\pi$. Any $\pi$ that achieves this minimum is said to \emph{testify} to, or be a \emph{witness} to, the value $K(\sigma|\tau)$. The \emph{Kolmogorov complexity} of a binary string $\sigma$ is $K(\sigma)=K(\sigma|\lambda)$, where $\lambda$ is the empty string.	These definitions extends naturally to other finite data objects, e.g., vectors in $\Q^n$, via standard binary encodings; see~\cite{LiVit08} for details.
	
	\subsection{Kolmogorov Complexity in Euclidean Spaces}\label{subsec:compEucPrelim}
	The above definitions can also be extended to Euclidean spaces, as we now describe. The \emph{Kolmogorov complexity} of a point $x\in\R^m$ at \emph{precision} $r\in\N$ is the length of the shortest program $\pi$ that outputs a \emph{precision}-$r$ rational estimate for $x$. Formally, it is
	\[K_r(x)=\min\left\{K(p)\,:\,p\in B_{2^{-r}}(x)\cap\Q^m\right\}\,,\]
	where $B_{\ve}(x)$ denotes the open ball of radius $\ve$ centered on $x$. The \emph{conditional Kolmogorov complexity of $x$ at precision $r$ given $y\in\R^n$ at precision $s\in\R^n$} is
	\[K_{r,s}(x|y)=\max\big\{\min\{K_r(p|q)\,:\,p\in B_{2^{-r}}(x)\cap\Q^m\}\,:\,q\in B_{2^{-s}}(y)\cap\Q^n\big\}\,.\]
	When the precisions $r$ and $s$ are equal, we abbreviate $K_{r,r}(x|y)$ by $K_r(x|y)$. As the following lemma shows, these quantities obey a chain rule and are only linearly sensitive to their precision parameters.
	
	\begin{lem}[J. Lutz and N. Lutz~\cite{LutLut17}, N. Lutz and Stull~\cite{LutStu17}]\label{lem:chain}
		Let $x \in \R^m$ and $y\in\R^n$. For all $r,s\in\N$ with $r\geq s$,
		\begin{enumerate}
			\item $K_r(x,y)=K_r(x|y)+K_r(y)+O(\log r)$.
			\item $K_r(x)=K_{r,s}(x|x)+K_s(x)+O(\log r)$.
		\end{enumerate}
	\end{lem}

	\subsection{Effective Hausdorff and Packing Dimensions}
	J. Lutz initiated the study of algorithmic dimensions by effectivizing Hausdorff dimension using betting strategies called~\emph{gales}, which generalize martingales. Subsequently, Athreya, et al., defined effective packing dimension, also using gales~\cite{AHLM07}. Mayordomo showed that effective Hausdorff dimension can be characterized using Kolmogorov complexity~\cite{Mayo02}, and Mayordomo and J. Lutz showed that effective packing dimension can also be characterized in this way~\cite{LutMay08}. In this paper, we use these characterizations as definitions.
	The \emph{effective Hausdorff dimension} and \emph{effective packing dimension} of a point $x\in\R^n$ are
	\[\dim(x)=\liminf_{r\to\infty}\frac{K_r(x)}{r}\quad\text{and}\quad\Dim(x) = \limsup_{r\to\infty}\frac{K_r(x)}{r}\,.\]
	Intuitively, these dimensions measure the density of algorithmic information in the point $x$. Guided by the information-theoretic nature of these characterizations, J. Lutz and N. Lutz~\cite{LutLut17} defined the \emph{lower} and \emph{upper conditional dimension} of $x\in\R^m$ given $y\in\R^n$ as
	\[\dim(x|y)=\liminf_{r\to\infty}\frac{K_r(x|y)}{r}\quad\text{and}\quad\Dim(x|y) = \limsup_{r\to\infty}\frac{K_r(x|y)}{r}\,.\]
	
	\subsection{Relative Complexity and Dimensions}
	By letting the underlying fixed prefix-free Turing machine $U$ be a universal \emph{oracle} machine, 
	we may \emph{relativize} the definition in this section to an arbitrary oracle set $A \subseteq \N$. The definitions of $K^A(\sigma|\tau)$, $K^A(\sigma)$, $K^A_r(x)$, $K^A_r(x|y)$, $\dim^A(x)$, $\Dim^A(x)$  $\dim^A(x|y)$, and $\Dim^A(x|y)$ are then all identical to their unrelativized versions, except that $U$ is given oracle access to $A$.
	
	We will frequently consider the complexity of a point $x \in \R^n$ \emph{relative to a point} $y \in \R^m$, i.e., relative to a set $A_y$ that encodes the binary expansion of $y$ is a standard way. We then write $K^y_r(x)$ for $K^{A_y}_r(x)$. J. Lutz and N. Lutz showed that $K_r^y(x)\leq K_{r,t}(x|y)+K(t)+O(1)$ \cite{LutLut17}.
	
	\section{Background and Approach}\label{sec:approach}
	
	In this section we describe the basic ideas behind our investigation of dimension spectra of lines. We briefly discuss some of our earlier work on this subject, and we present two technical lemmas needed for the proof our main theorems.
	
	The dimension of a point on a line in $\R^2$ has the following trivial bound.
	\begin{obs}\label{obs:upper}
		For all $a,b,x\in\R$, $\dim(x,ax+b)\leq\dim(x,a,b)$.
	\end{obs}
	In this work, our goal is to find values of $x$ for which the approximate converse
	\begin{equation}\label{eq:lower}
	\dim(x,ax+b)\geq\dim^{a,b}(x)+\dim(a,b)
	\end{equation}
	holds. There exist oracles, at least, relative to which (\ref{eq:lower}) does not always hold. This follows from the point-to-set principle of J. Lutz and N. Lutz~\cite{LutLut17} and the existence of Furstenberg sets with parameter $\alpha$ and Hausdorff dimension less than $1+\alpha$ (attributed by Wolff~\cite{Wolf99} to Furstenberg and Katznelson ``in all probability''). The argument is simple and very similar to our proof in~\cite{LutStu17} of a lower bound on the dimension of generalized Furstenberg sets.
	
	Specifically, for every $s\in[0,1]$, we want to find an $x$ of effective Hausdorff dimension $s$ such that (\ref{eq:lower}) holds. Note that equality in Observation~\ref{obs:upper} implies (\ref{eq:lower}). 
	\begin{obs}
		Suppose $ax+b=ux+v$ and $u\neq a$. Then \[\dim(u,v)\geq\dim^{a,b}(u,v)\geq\dim^{a,b}\left(\frac{b-v}{u-a}\right)=\dim^{a,b}(x) \,.\]
	\end{obs}
	This observation suggests an approach, whenever $\dim^{a,b}(x)>\dim(a,b)$, for showing that $\dim(x,ax+b)\geq\dim(x,a,b)$. Since $(a,b)$ is, in this case, the unique low-dimensional pair such that $(x,ax+b)$ lies on $L_{a,b}$, one might na\"ively hope to use this fact to derive an estimate of $(x,a,b)$ from an estimate of $(x,ax+b)$. Unfortunately, the dimension of a point is not even semicomputable, so algorithmically distinguishing $(a,b)$ requires a more refined statement.
	\subsection{Previous Work}
	The following lemma, which is essentially geometrical, is such a statement.
	\begin{lem}[N. Lutz and Stull~\cite{LutStu17}]\label{lem:lines}
		Let $a,b,x\in\R$. For all $(u,v)\in\R^2$ such that $u x+v=ax+b$ and $t=-\log\|(a,b)-(u,v)\|\in(0,r]$,
		\[K_{r}(u,v)\geq K_t(a,b) + K^{a,b}_{r-t}(x)-O(\log r)\,.\]
	\end{lem}
	
	Roughly, if $\dim(a,b)<\dim^{a,b}(x)$, then Lemma~\ref{lem:lines} tells us that $K_r(u,v)>K_r(a,b)$ unless $(u,v)$ is very close to $(a,b)$. As $K_r(u,v)$ is upper semicomputable, this is algorithmically useful: We can enumerate all pairs $(u,v)$ whose precision-$r$ complexity falls below a certain threshold. If one of these pairs satisfies, approximately, $ux+v=ax+b$, then we know that $(u,v)$ is close to $(a,b)$. Thus, an estimate for $(x,ax+b)$ algorithmically yields an estimate for $(x,a,b)$.
	
	In our previous work~\cite{LutStu17}, we used an argument of this type to prove a general lower bound on the dimension of points on lines in $\R^2$:
	\begin{thm}[N. Lutz and Stull~\cite{LutStu17}]\label{thm:boundingMainThm}
		For all $a,b,x\in\R$,
		\begin{align*}
		\dim(x,ax+b)\geq \dim^{a,b}(x)+ \min\{\dim(a,b),\,\dim^{a,b}(x)\}\,.
		\end{align*}
	\end{thm}
	The strategy in that work is to use oracles to artificially lower $K_r(a,b)$ when necessary, to essentially force $\dim(a,b)<\dim^{a,b}(x)$. This enables the above argument structure to be used, but lowering the complexity of $(a,b)$ also weakens the conclusion, leading to the minimum in Theorem~\ref{thm:boundingMainThm}.
	
	\subsection{Technical Lemmas}
	In the present work, we circumvent this limitation and achieve inequality (\ref{eq:lower}) by controlling the choice of $x$ and placing a condition on $(a,b)$. Adapting the above argument to the case where $\dim(a,b)>\dim^{a,b}(x)$ requires refining the techniques of~\cite{LutStu17}. In particular, we use the following two technical lemmas, which strengthen results from that work. Lemma~\ref{lem:pointInduct} weakens the conditions needed to compute an estimate of $(x,a,b)$ from an estimate of $(x,ax+b)$.
	
	\begin{lem}\label{lem:pointInduct}
		Let $a,b,x\in\R$, $k \in \N$, and $r_0=1$. Suppose that $r_1,\ldots, r_k\in\N$, $\delta\in\R_+$, and $\ve,\eta\in\Q_+$ satisfy the following conditions for every $1\leq i\leq k$.
		\begin{enumerate}
			\item $r_i \geq \log(2|a|+|x|+6)+r_{i-1}$.
			\item $K_{r_i}(a,b)\leq \left(\eta+\ve\right)r_i$.
			\item For every $(u,v)\in\R^2$ such that $t=-\log\|(a,b)-(u,v)\|\in(r_{i-1},r_i]$ and $ux+v=ax+b$, $K_{r_i}(u,v)\geq\left(\eta-\ve\right)r_i+\delta\cdot(r_i- t)$.
		\end{enumerate}
		Then for every oracle set $A \subseteq \N$, 
		\[K^A_{r_k}(a, b, x \, | \, x, ax + b) \leq 2^{k}\left(K(\ve)+ K(\eta) + \frac{4\ve}{\delta} r_k + O(\log r_k)\right)\,.\]
	\end{lem}
	Lemma~\ref{lem:oraclesSpectra} strengthens the oracle construction of~\cite{LutStu17}, allowing us to control complexity at multiple levels of precision.
	\begin{lem}\label{lem:oraclesSpectra}
		Let $z\in\R^n$, $\eta\in\Q\cap[0,\dim(z)]$, and $k\in\N$. For all $r_1, \ldots, r_k \in \N$, there is an oracle $D=D(r_1,\ldots, r_k,z,\eta)$ such that
		\begin{enumerate}
			\item For every $t \leq r_1$,  $K^D_t(z) =\min\{\eta r_1,K_t(z)\}+ O(\log r_k)$
			\item For every $1 \leq i \leq k$, \[K^D_{r_i}(z) = \eta r_1 + \sum_{j =2}^i \min\{\eta (r_j - r_{j-1}), K_{r_j, r_{j-1}}(z \, | \, z)\} + O(\log r_k)\,.\]
			\item For every $t\in\N$ and $x\in\R$, $K^{z,D}_t(x) = K^z_t(x) + O(\log r_k)$.
		\end{enumerate}
	\end{lem}
	
	\section{Main Theorems}\label{ssec:mainthm}
	We are now prepared to prove our two main theorems. We first show that, for lines $L_{a, b}$ such that $\dim(a, b) = \Dim(a, b)$, the dimension spectrum $\spec(L_{a,b})$ contains the unit interval.
	\begin{thm}\label{thm:spectraMainTheorem}
		Let $a, b \in\R$ satisfy $\dim(a, b) = \Dim(a, b)$. Then for every $s \in [0, 1]$ there is a point $x\in\R$ such that $\dim(x, ax + b) = s + \min\{\dim(a,b), 1\}$.
	\end{thm}
	
	\begin{proof}
		Every line contains a point of effective Hausdorff dimension 1~\cite{Ture11}, and by the preservation of effective dimensions under computable bi-Lipschitz functions, $\dim(a,a^2+b)=\dim(a,b)$, so the theorem holds for $s=0$.
		
		Now let $s \in (0,1]$ and $d=\dim(a, b) = \Dim(a, b)$. Let $y \in \R$ be random relative to $(a, b)$. That is, there is some constant $c \in \N$ such that for all $r \in \N$, $K^{a, b}_r(y) \geq r - c$.
		Define sequence of natural numbers $\{h_j\}_{j \in \N}$ inductively as follows. Define $h_0 = 1$. For every $j > 0$, define 
		\begin{equation*}
		h_j = \min\left\{h \geq 2^{h_{j-1}}: K_h(a, b) \leq \left(d + \frac{1}{j}\right)h\right\}.
		\end{equation*}
		Note that $h_j$ always exists. For every $r \in \N$, let
		\begin{align*}
		x[r] = \begin{cases}
		0 &\text{ if } \frac{r}{h_j} \in (s, 1] \text{ for some } j \in \N \\
		y[r] &\text{ otherwise}
		\end{cases}
		\end{align*}
		Define $x \in \R$ to be the real number with this binary expansion. Then $K_{sh_j}(x)=sh_j+O(\log sh_j)$.
		
		We first show that $\dim(x, ax + b) \leq s + \min\{d, 1\}$. For every $j \in \N$, 
		\begin{align*}
		K_{h_j}(x,ax+b)&= K_{h_j}(x) + K_{h_j}(ax + b \, | \, x) + O(\log h_j) \\
		&= K_{sh_j}(x) + K_{h_j}(ax + b \, | \, x) + O(\log h_j) \\
		&= K_{sh_j }(y) + K_{h_j}(ax + b \, | \, x) + O(\log h_j) \\
		&\leq sh_j + \min\{d,1\}\cdot h_j + o(h_j)\,.
		\end{align*}
		Therefore, 
		\begin{align*}
		\dim(x, ax + b) &= \liminf_{r \to \infty} \frac{K_r(x, ax + b)}{r}\\
		&\leq \liminf_{j \to\infty} \frac{K_{h_j}(x, ax + b)}{h_j}\\
		&\leq \liminf_{j \to \infty} \frac{s h_j + \min\{d, 1\} h_j  +o(h_j)}{h_j}\\
		&= s + \min\{d, 1\}\,.
		\end{align*}
		
		If $1\geq s \geq d$, then by Theorem~\ref{thm:boundingMainThm} we also have
		\begin{align*}
		\dim(x,ax+b)&\geq \dim(x\,|\,a,b) +\dim(a,b)\\
		&=\dim(x)+d\\
		&=\liminf_{r\to\infty}\frac{K_r(x)}{r}+d\\
		&=\liminf_{j\to\infty}\frac{K_{h_j}(x)}{h_j}+d\\
		&=s+\min\{d,1\}\,.
		\end{align*}
		Hence, we may assume that $s < d$.
		
		Let $H = \Q \cap (s, \min\{d,1\})$. We now show that for every $\eta \in H$ and $\ve \in \Q_+$, $\dim(x, ax + b) \geq s + \eta - \alpha\ve$, where $\alpha$ is some constant independent of $\eta$ and $\ve$.
		
		Let $\eta \in H$, $\delta = 1 - \eta > 0$, and $\ve \in \Q_+$. Let $j \in \N$ and $m = \frac{s - 1}{\eta - 1}$. We first show that 
		\begin{equation}\label{eq:mainThmeq}
		K_{r}(x, ax + b) \geq K_r(x) + \eta r - c\frac{\ve}{\delta} r - o(r)\,,
		\end{equation}
		for every $r \in (sh_j, mh_j]$. Let $r \in (sh_j, mh_j]$. Set $k = \frac{r}{sh_j} $, and define $r_i = i s h_j$ for all $1 \leq i \leq k$. Note that $k$ is bounded by a constant depending only on $s$ and $\eta$. Therefore $o(r_k)$ is sublinear for all $r_i$. Let $D_{r} = D(r_1,\ldots, r_k, (a, b), \eta)$ be the oracle defined in Lemma \ref{lem:oraclesSpectra}. We first note that, since $\dim(a, b) = \Dim(a, b)$,
		\begin{align*}
		K_{r_i, r_{i-1}}(a, b \, | \, a, b) &= K_{r_i}(a, b) - K_{r_{i-1}}(a, b) - O(\log r_i)\\
		&= \dim(a, b)r_i - o(r_i) - \dim(a, b) r_{i-1} - o(r_{i-1}) - O(\log r_i)\\
		&= \dim(a, b)(r_i - r_{i-1}) - o(r_i)\\
		&\geq \eta(r_i - r_{i-1}) - o(r_i).
		\end{align*}
		Hence, by property 2 of Lemma~\ref{lem:oraclesSpectra}, for every $1 \leq i \leq k$,
		\begin{equation}\label{eq:boundOracle}
		\vert K^{D_r}_{r_i}(a, b) - \eta r_i \vert \leq o(r_k).
		\end{equation}
		We now show that the conditions of Lemma~\ref{lem:pointInduct} are satisfied. By inequality (\ref{eq:boundOracle}), for every $1 \leq i \leq k$, 
		\[K^{D_r}_{r_i}(a, b) \leq \eta r_i + o(r_k)\,,\]
		and so $K^{D_r}_{r_i}(a, b) \leq (\eta + \ve) r_i$, for sufficiently large $j$. Hence, condition 2 of Lemma \ref{lem:pointInduct} is satisfied. 
		
		To see that condition 3 is satisfied for $i=1$, let $(u, v) \in B_1(a, b)$ such that $ux + v = ax + b$ and $t=-\log\|(a,b)-(u,v)\| \leq r_1$. Then, by Lemmas \ref{lem:lines} and \ref{lem:oraclesSpectra}, and our construction of $x$,
		\begin{align*}
		K^{D_{r}}_{r_1}(u,v) &\geq K^{D_{r}}_t(a,b) + K^{D_{r}}_{r_1-t,r_1}(x|a,b)-O(\log r_1)\\
		&\geq \min\{\eta r_1, K_t(a, b)\} + K_{r_1-t}(x)- o(r_k)\\
		&\geq \min\{\eta r_1, dt - o(t)\} + (\eta + \delta)(r_1 - t) - o(r_k)\\ 
		&\geq \min\{\eta r_1, \eta t - o(t)\} + (\eta + \delta)(r_1 - t) - o(r_k) \\ 
		&\geq \eta t - o(t) + (\eta + \delta)(r_1 - t) - o(r_k)\,.
		\end{align*}
		We conclude that $K^{D_{r}}_{r_1}(u,v)  \geq (\eta - \ve)r_1 + \delta(r_1 - t)$, for all sufficiently large $j$.
		
		To see that that condition 3 is satisfied for $1 < i \leq k$, let $(u, v) \in B_{2^{-r_{i-1}}}(a, b)$ such that $ux + v = ax + b$ and $t=-\log\|(a,b)-(u,v)\| \leq r_i$. Since $(u, v) \in B_{2^{-r_{i-1}}}(a, b)$, 
		\[r_i - t \leq r_i - r_{i-1}= ish_j - (i-1)sh_j \leq  sh_j  + 1\leq r_1 + 1\,.\]
		Therefore, by Lemma \ref{lem:lines}, inequality (\ref{eq:boundOracle}), and our construction of $x$,
		\begin{align*}
		K^{D_{r}}_{r_i}(u,v) &\geq K^{D_{r}}_t(a,b) + K^{D_{r}}_{r_i-t,r_i}(x|a,b)-O(\log r_i)\\
		&\geq \min\{\eta r_i, K_t(a, b)\} + K_{r_i-t}(x)- o(r_i)\\
		&\geq \min\{\eta r_i, dt - o(t)\} + (\eta + \delta)(r_i - t) - o(r_i)\\ 
		&\geq \min\{\eta r_i, \eta t - o(t)\} + (\eta + \delta)(r_i - t) - o(r_i) \\ 
		&\geq \eta t - o(t) + (\eta + \delta)(r_i - t) - o(r_i)\,,
		\end{align*}
		We conclude that $K^{D_{r}}_{r_i}(u,v)  \geq (\eta - \ve)r_i + \delta(r_i - t)$, for all sufficiently large $j$. Hence the conditions of Lemma \ref{lem:pointInduct} are satisfied, and we have
		\begin{align*}
		K_{r}(x, ax + b) \geq&\ K^{D_{r}}_{r}(x, ax + b) - O(1)\\
		\geq&\ K^{D_{r}}_{r}(a, b, x) - 2^k\left(K(\ve) + K(\eta) + \frac{4\ve}{\delta} r + O(\log r)\right)\\
		=&\ K^{D_{r}}_{r}(a, b) + K^{D_{r}}_{r}(x \, | \, a, b) \\
		&\qquad\qquad\ \,- 2^k\left(K(\ve) + K(\eta) + \frac{4\ve}{\delta} r + O(\log r)\right)\\
		\geq&\ sr + \eta r - 2^k\left(K(\ve) + K(\eta) + \frac{4\ve}{\delta} r + O(\log r)\right).
		\end{align*}
		Thus, for every $r \in (sh_j, mh_j]$,
		\begin{equation*}
		K_{r}(x, ax + b) \geq sr + \eta r - \frac{\alpha\ve}{\delta} r - o(r)\,,
		\end{equation*}
		where $\alpha$ is a fixed constant, not depending on $\eta$ and $\ve$.

		To complete the proof, we show that (\ref{eq:mainThmeq}) holds for every $r \in [mh_j, sh_{j + 1})$. By Lemma~\ref{lem:chain} and our construction of $x$,
		\begin{align*}
		K_r(x) &= K_{r, h_j}(x\,|\,x) + K_{h_j}(x)+o(r)
		\\&= r - h_j  + sh_j+o(r)
		\\&\geq \eta r+o(r)\,.
		\end{align*}
		The proof of Theorem \ref{thm:boundingMainThm} gives $K_r(x, ax + b) \geq K_r(x) + \dim(x)r - o(r)$, and so $K_r(x, ax + b) \geq r(s + \eta)$.
		
		Therefore, equation (\ref{eq:mainThmeq}) holds for every $r \in [sh_j, sh_{j+1})$, for all sufficiently large $j$. Hence,
		\begin{align*}
		\dim(x, ax + b) &= \liminf\limits_{r \rightarrow \infty} \frac{K_r(x, ax + b)}{r}\\
		&\geq \liminf\limits_{r \rightarrow \infty} \frac{K_r(x) + \eta r - \frac{\alpha\ve}{\delta} r - o(r)}{r}\\
		&\geq \liminf\limits_{r \rightarrow \infty} \frac{K_r(x)}{r} + \eta - \frac{\alpha\ve}{\delta}\\
		&= s + \eta - \frac{\alpha\ve}{\delta}\,.
		\end{align*}
		Since $\eta$ and $\ve$ were chosen arbitrarily, the conclusion follows.
	\end{proof}

	\begin{thm}\label{thm:spectraMainThm2}
		Let $a, b \in\R$ such that $\dim(a, b) \geq 1$. Then for every $s \in [\frac{1}{2}, 1]$ there is a point $x\in\R$ such that  $\dim(x, ax + b) \in \left[\frac{3}{2} + s - \frac{1}{2s}, s + 1\right]$.
	\end{thm}
	
	\begin{cor}
		Let $L_{a, b}$ be any line in $\R^2$. Then the dimension spectrum $\spec(L_{a, b})$ is infinite.
	\end{cor}
	\begin{proof}
		Let $(a, b) \in R^2$. If $\dim(a, b) < 1$, then by Theorem \ref{thm:boundingMainThm} and Observation \ref{obs:upper}, the spectrum $\spec(L_{a, b})$ contains the interval $[\dim(a, b) , 1]$. Assume that $\dim(a, b) \geq 1$. By Theorem \ref{thm:spectraMainThm2}, for every $s \in [\frac{1}{2}, 1]$, there is a point $x$ such that $\dim(x, ax + b) \in [\frac{3}{2} + s - \frac{1}{2s}, s + 1]$. Since these intervals are disjoint for $s_n = \frac{2n - 1}{2n}$, the dimension spectrum $\spec(L_{a, b})$ is infinite.
	\end{proof}
	\section{Future Directions}\label{sec:conc}
	
	We have made progress in the broader program of describing the dimension spectra of lines in Euclidean spaces. We highlight three specific directions for further progress. First, it is natural to ask whether the condition on $(a,b)$ may be dropped from the statement our main theorem: \emph{Does Theorem~\ref{thm:spectraMainTheorem} hold for arbitrary} $a,b\in\R$?

	Second, the dimension spectrum of a line $L_{a,b}\subseteq\R^2$ may \emph{properly} contain the unit interval described in our main theorem, even when $\dim(a,b)=\Dim(a,b)$. If $a\in\R$ is random and $b=0$, for example, then $\spec(L_{a,b})=\{0\}\cup[1,2]$. It is less clear whether this set of ``exceptional values'' in $\spec(L_{a,b})$ might itself contain an interval, or even be infinite. \emph{How large (in the sense of cardinality, dimension, or measure) may $\spec(L_{a,b})\cap\big[0,\min\{1,\dim(a,b)\}\big)$ be?}
	
	Finally, any non-trivial statement about the dimension spectra of lines in higher-dimensional Euclidean spaces would be very interesting. Indeed, an $n$-dimensional version of Theorem~\ref{thm:boundingMainThm} (i.e., one in which $a,b\in\R^{n-1}$, for all $n\geq 2$) would, via the point-to-set principle for Hausdorff dimension~\cite{LutLut17}, affirm the famous Kakeya conjecture and is therefore likely difficult. The additional hypothesis of Theorem~\ref{thm:spectraMainTheorem} might make it more conducive to such an extension. 
	
	\bibliography{DSPL}
	\clearpage
	\pagenumbering{arabic}
	\renewcommand*{\thepage}{A\arabic{page}}
	\appendix
	
	\section{Technical Appendix}\label{app:main}
	\renewcommand\thethm{\thesection.\arabic{thm}}
	\setcounter{thm}{0}
	
	\subsection{Precursors to Technical Lemmas}\label{sec:lemmas}
	The following lemma will be used in the proof of Lemma \ref{lem:pointInduct}.
	\begin{lem}[Case and J. Lutz~\cite{CasLut15}, J. Lutz and N. Lutz~\cite{LutLut17}]\label{lem:sensitivity}
		Let $x\in\R^m$ and $y\in\R^n$. For all $r,s,r',s'\in\N$,
		\begin{enumerate}
			\item $K_{r'}(x)=K_r(x)+O(|r'-r|)+O(\log r)$.
			\item $K_{r',s'}(x|y)=K_{r,s}(x|y)+O(|r'-r|+|s'-s|)+O(\log rs)$.
		\end{enumerate}
	\end{lem}
	The following two lemmas from our previous work (stated in slightly different forms here) are precursors to Lemmas~\ref{lem:pointInduct} and~\ref{lem:oraclesSpectra}. The proof of Lemma~\ref{lem:pointInduct} is similar to that of Lemma~\ref{lem:pointBounding}, and the proof of Lemma~\ref{lem:oraclesSpectra} is an induction on Lemma~\ref{lem:oracles}.
	\begin{lem}[N. Lutz and Stull~\cite{LutStu17}]\label{lem:pointBounding}
		Suppose that $a,b,x\in\R$, $r\in\N$, $\delta\in\R_+$, and $\ve,\eta\in\Q_+$ satisfy the following conditions.
		\begin{enumerate}
			\item $r\geq \log(2|a|+|x|+5)+1$.
			\item $K_r(a,b)\leq \left(\eta+\ve\right)r$.
			\item For every $(u,v)\in \R^2$ such that $t=-\log\|(a,b)-(u,v)\|\in(0,r]$ and $ux+v=ax+b$, $K_{r}(u,v)\geq\left(\eta-\ve\right)r+\delta\cdot(r- t)$.
		\end{enumerate}
		Then for every oracle set $A \subseteq \N$, 
		\[K^A_{r}(x, ax + b) \geq K^A_r(a, b, x) - \frac{4\ve}{\delta}r - K(\ve) - K(\eta) - O(\log r)\,.\]
	\end{lem}
	
	\begin{lem}[N. Lutz and Stull~\cite{LutStu17}]\label{lem:oracles}
		Let $r\in\N$, $z\in\R^n$,  and $\eta\in\Q\cap[0,\dim(z)]$. There is an oracle $A=A(r,z,\eta)$ such that
		\begin{enumerate}
			\item For every $t\leq r$, $K^{A}_t(z)=\min\{\eta r,K_t(z)\}+O(\log r)$.
			\item For every $t > r$, $K^{A}_t(z) \geq \eta r + K_{t, r}(z \, | \, z) + O(\log r)$.
			\item For every $t\in\N$ and $y\in\R^m$, $K_t^{z,A}(y)=K^{z}_t(y)+O(\log r)$.
		\end{enumerate}
	\end{lem}
	
	\subsection{Computing a Line Given a Point}\label{ssec:computeLine}
	For our purposes, we will need the following corollary to Lemma \ref{lem:pointBounding}. Informally, that lemma gives conditions under which precision-$r$ estimates for $(x,ax+b)$ and $(a,b,x)$ contain similar amounts of information. This corollary shows that, under the same conditions, those two approximations are furthermore nearly ``interchangeable,'' in the sense that there is a short program which, given a precision-$r$ estimate for $(x,ax+b)$ as input, will output a precision-$r$ estimate for $(a, b, x)$, and, as we argue in the proof, vice versa.
	\begin{cor}\label{cor:point}
		If the conditions of Lemma \ref{lem:pointBounding} are satisfied, then
		\[K^A_r(a,b,x|x,ax+b)\leq\frac{4\ve}{\delta}r + K(\ve) + K(\eta) + O(\log r)\,.\]
	\end{cor}
	\begin{proof}
		It is easy to see that $K_r(x,ax+b|a,b,x)=O(\log r)$: consider a constant-length program that, given $(u,v,y)\in\Q^3$, outputs $(y,uy+v)$. If $(u,v,y)\in B_{2^{-r}}(a,b,x)$, then $(y,uy+v)\in B_{2^{c-r}}(x,ax+b)$, where $c$ is constant in $r$, so $K_{r-c,r}(ax+b|a,b,x)=O(1)$. Thus, by Lemma~\ref{lem:sensitivity}, $K_{r}(ax+b|a,b,x)=O(\log r)$.
		
		Now suppose that the conditions of Lemma 6 are satisfied. Then by symmetry of information and Lemma \ref{lem:pointBounding},
		\begin{align*}
		K^A_r(a,b,x|x,ax+b)&=K^A_r(a,b,x)-K^A_r(x,ax+b)+K_r^A(x,ax+b|a,b,x)\\
		&=K^A_r(a,b,x)-K^A_r(x,ax+b)+O(\log r)\\
		&\leq\frac{4\ve}{\delta}r + K(\ve) + K(\eta) + O(\log r)\,.
		\end{align*}
	\end{proof}
	
	We will also need the following pair of geometric facts.
	\begin{obs}[N. Lutz and Stull~\cite{LutStu17}]\label{obs:linemachine}
		Let $a,x,b\in\R$ and $r\in\N$. Let $(q_1,q_2)\in B_{2^{-r}}(x,ax+b)$.
		\begin{enumerate}
			\item If $(p_1,p_2)\in B_{2^{-r}}(a,b)$, then $|p_1q_1+p_2-q_2|< 2^{-r}(|p_1|+|q_1|+3)$.
			\item If $|p_1q_1+p_2-q_2|\leq 2^{-r}(|p_1|+|q_1|+3)$, then there is some pair $(u,v)\in B_{2^{-r}(2|a|+|x|+5)}(p_1,p_2)$ such that $ax+b=ux+v$.
		\end{enumerate}
	\end{obs}
	{\renewcommand{\thethm}{\ref{lem:pointInduct}}
		\begin{lem}
			Let $a,b,x\in\R$, $k \in \N$, and $r_0=1$. Suppose that $r_1,\ldots, r_k\in\N$, $\delta\in\R_+$, and $\ve,\eta\in\Q_+$ satisfy the following conditions for every $1\leq i\leq k$.
			\begin{enumerate}
				\item $r_i \geq \log(2|a|+|x|+6)+r_{i-1}$.
				\item $K_{r_i}(a,b)\leq \left(\eta+\ve\right)r_i$.
				\item For every $(u,v)\in\R^2$ such that $t=-\log\|(a,b)-(u,v)\|\in(r_{i-1},r_i]$ and $ux+v=ax+b$, $K_{r_i}(u,v)\geq\left(\eta-\ve\right)r_i+\delta\cdot(r_i- t)$.
			\end{enumerate}
			Then for every oracle set $A \subseteq \N$, 
			\[K^A_{r_k}(a, b, x \, | \, x, ax + b) \leq 2^{k}\left(K(\ve)+ K(\eta) + \frac{4\ve}{\delta} r_k + O(\log r_k)\right)\,.\]
		\end{lem}
		\addtocounter{thm}{-1}
	}
	\begin{proof}
		Let $a,b,x\in\R$. We proceed by induction on $k$. By Corollary \ref{cor:point}, the conclusion holds for $k = 1$. Assume the conclusion holds for all $i < k$. Let $r_1, \ldots, r_k$, $\delta$, $\ve$, $\eta$, and $A$ be as described in the lemma statement. 
		
		Define an oracle Turing machine $M$ that does the following given oracle $A$ and input $\pi=\pi_1\pi_2\pi_3\pi_4\pi_5$ such that  $U^A(\pi_1)=(q_1,q_2)\in\Q^2$, $U(\pi_2)= (s_1, \ldots, s_k) \in\N^k$,  $U(\pi_3)=\zeta\in\Q$, $U(\pi_4)=\iota\in\Q$ and $U^A(\pi_5, q_1, q_2)=h\in\Q^2$
		
		For every program $\sigma\in\{0,1\}^*$ with $\ell(\sigma)\leq (\iota+\zeta)s_k$, in parallel, $M$ simulates $U(\sigma)$. If one of the simulations halts with some output $(p_1,p_2)\in \Q^2\cap B_{2^{-r_{k-1}}}(h)$ such that 
		\[|p_1q_1+p_2-q_2|< 2^{-s_2}(|p_1|+|q_1|+3)\,,\]
		then $M$ halts with output $(p_1,p_2,q_1)$. Let $c_M$ be a constant for the description of $M$.
		
		Now let $\pi_1$, $\pi_2$, $\pi_3$, $\pi_4$, and $\pi_5$ testify to $K^A_r(x,ax+b)$, $K(r_1, \ldots, r_k)$, $K(\ve)$, $K(\eta)$, and $K_{r_{k-1}, r_k}(a,b \, | \, x, ax + b)$ respectively, and let $\pi=\pi_1\pi_2\pi_3\pi_4\pi_5$.
		
		By condition 2, there is some $(\hat{p}_1,\hat{p}_2)\in B_{2^{-r_k}}(a,b)$ such that $K(\hat{p}_1,\hat{p}_2)\leq (\eta+\ve)r_k$, meaning that there is some $\hat{\sigma}\in\{0,1\}^*$ with $\ell(\hat{\sigma})\leq(\eta+\ve)r_k$ and $U(\hat{\sigma})=(\hat{p}_1,\hat{p}_2)$. By Observation~\ref{obs:linemachine}(1),
		\[|\hat{p}_1q_1+\hat{p}_2-q_2|< 2^{-r_k}(|\hat{p}_1|+|q_1|+3)\,,\] 
		for every $(q_1,q_2)\in B_{2^{-r_k}}(x,ax+b)$, so $M$ is guaranteed to halt on input $\pi$.
		
		Hence, let $(p_1,p_2,q_1)=M(\pi)$. By Observation \ref{obs:linemachine}(2), there is some
		\[(u,v)\in B_{2^{\gamma-r_k}}(p_1,p_2)\subseteq B_{2^{-r_{k-1}}}(a,b)\] 
		such that $ux+v=ax+b$, where $\gamma=\log(2|a|+|x|+5)$. We have \[\|(p_1,p_2)-(u,v)\|<2^{\gamma-r_k}\]
		and $|q_1-x|<2^{-r_k}$, so 
		\[(p_1,p_2,q_1)\in B_{2^{\gamma+1-r_k}}(u,v,x)\,.\]
		It therefore follows that
		\begin{align*}
		K^A_{r_k-\gamma-1, r_k}(u,v,x \, | \, x, ax + b )&\leq K(p_1, p_2, q_1)\\
		&\leq \ell(\pi_1\pi_2\pi_3\pi_4\pi_5)+c_M\\
		&\leq \ell(\pi_5) + K(r_1, \ldots, r_k) +K(\ve)+K(\eta)+c_M\\
		&= \ell(\pi_5) + K(\ve)+ K(\eta)+ O(\log r_k)\,.
		\end{align*}
		
		Applying Lemma~\ref{lem:sensitivity} yields
		\begin{equation}\label{eq:computingMain}
		K^A_{r_k}(u,v,x \, | \, x, ax + b ) \leq \ell(\pi_5) + K(\ve)+ K(\eta)+ O(\log r_k).
		\end{equation}
		By our inductive hypothesis, we have that
		\begin{align}\label{eq:lengthPi5}
		\ell(\pi_5) &= K_{r_{k-1}, r_k}(a,b \, | \, x, ax + b) \nonumber \\
		&= K_{r_{k-1}}(a,b \, | \, x, ax + b) + O(\log r_{k-1})\nonumber \\
		&\leq 2^{k-1}\left(K(\ve)+ K(\eta) + \frac{4\ve}{\delta} r_{k-1} + O(\log r_{k-1})\right)\,.
		\end{align}
		
		To complete the proof, we bound $K^A_{r_k}(a, b, x \, | \, u, v, x)$. If $t > r_k$, then 
		\[K^A_{r_k}(a, b, x \, | \, u, v, x) \leq \log(r_k)\,.\]
		Otherwise, when $t\leq r_k$, by our construction of $M$ and Lemma~\ref{lem:sensitivity},
		\begin{align*}
		(\eta+\ve)r_k&\geq K(p_1,p_2)\\
		&\geq K_{r_k-\gamma}(u,v)\\
		&\geq K_{r_k}(u,v)-O(\log r_k)\,.
		\end{align*}
		Combining this with condition 3 in the lemma statement and simplifying yields 
		\[r_k-t\leq \frac{2\ve}{\delta}r_k+O(\log r_k)\,.\] 
		
		Therefore, by Lemma~\ref{lem:sensitivity}, we have
		\begin{align}\label{eq:abxFromuvx}
		K_{r_k}(a, b, x \, | \, u, v, x) &\leq 2(r_k - t) + O(\log r_k) \nonumber \\
		&\leq \frac{4\ve}{\delta} r_k + O(\log r_k)\,,
		\end{align}
		for every $t \in \N$.

		Combining inequalities (\ref{eq:computingMain}), (\ref{eq:lengthPi5}) and (\ref{eq:abxFromuvx}) gives
		\begin{align*}
		K_{r_k}(a, b, x \, | \, x, ax + b) &\leq  K_{r_k}(u, v, x \, | \, x, ax + b) + K_{r_k}(a, b, x \, | \, u, v, x)\\
		&\leq  K_{r_k}(u, v, x \, | \, x, ax + b) + \frac{4\ve}{\delta} r_k + O(\log r_k)\\
		&\leq  \ell(\pi_5) + K(\ve)+ K(\eta) + \frac{4\ve}{\delta} r_k + O(\log r_k)\\
		&\leq  2^k \left(K(\ve)+ K(\eta) + \frac{4\ve}{\delta} r_k + O(\log r_k)\right)\,.
		\end{align*}
	\end{proof}
	
	\subsection{Decreasing Complexity Using an Oracle}\label{ssec:oracles}
	
	Given an oracle $D\subseteq\N$, $r\in\N$, $z\in\R^n$, and $\eta\in\Q\cap[0,\dim^D(z)]$, let $A^D(r,z,\eta)$ be the oracle guaranteed by applying Lemma~\ref{lem:oracles} relative to $D$. For oracles, $A,B\subseteq\N$, let $\langle A,B\rangle\subseteq\N$ be an oracle that combines $A$ and $B$ by interleaving. Note that we treat $k$ as a constant for the purposes of asymptotic notation.
	{\renewcommand{\thethm}{\ref{lem:oraclesSpectra}}
		\begin{lem}
			Let $z\in\R^n$, $\eta\in\Q\cap[0,\dim(z)]$, and $k\in\N$. For all $r_1, \ldots, r_k \in \N$, there is an oracle $D=D(r_1,\ldots, r_k,z,\eta)$ such that
			\begin{enumerate}
				\item For every $t \leq r_1$,  $K^D_t(z) =\min\{\eta r_1,K_t(z)\}+ O(\log r_k)$
				\item For every $1 \leq i \leq k$, \[K^D_{r_i}(z) = \eta r_1 + \sum_{l =2}^i \min\{\eta (r_l - r_{l-1}), K_{r_l, r_{l-1}}(z \, | \, z)\} + O(\log r_k)\,.\]
				\item For every $t\in\N$ and $x\in\R$, $K^{z,D}_t(x) = K^z_t(x) + O(\log r_k)$.
			\end{enumerate}
		\end{lem}
		\addtocounter{thm}{-1}}
	\begin{proof}
		We define the sequence of oracles recursively. Let $D_1=A(r_1,z,\eta)$, as defined in Lemma~\ref{lem:oracles}, and for every $1<i\leq k$, let
		\[D_i=\left\{\begin{array}{ll}
		D_{i-1}&\text{if }K_{r_i}^{D_{i-1}}(z)<\eta r_i\\
		\langle D_{i-1},A^{D_{i-1}}(r_i,z,\eta)\rangle&\text{otherwise}\,.
		\end{array}\right.\]
		Notice that, for every $1\leq i\leq k$, $D_{i}$ is a finite oracle, so $\dim^{D_{i}}(z)=\dim(z)$ and $\eta\in[0,\dim^{D_k}(z)]$.
		
		We now show via induction on $k$ that the lemma holds for all $k\in\N$. For $k = 1$, all three properties hold by Lemma~\ref{lem:oracles}. Fix $j>1$, assume the properties hold for $k=j-1$.
		
		We first show that property 1 holds for $k=j$. Let $t \leq r_1$. It follows from the definition of the oracle $D_j$ and Lemma~\ref{lem:oracles}, relative to $D_{j-1}$, that
		\[K_t^{D_j}(z) = \min\{\eta r_j, K_t^{D_{j-1}}(z)\}+O(\log r_j)\,.\]
		By the induction hypothesis,
		$K^{D_{j-1}}_t(z)=\min\{\eta r_1,K_t(z)\} + O(\log r_{j-1})$. Thus,
		\begin{align*}
		K_t^{D_j}(z)&=\min\{\eta r_j, \min\{\eta r_1,K_t(z)\}+O(\log r_{j-1})\}+O(\log r_j)
		\\&=\min\{\eta r_1,K_t(z)\}+O(\log r_j)\,.
		\end{align*}
		
		We now show the property 2 holds for $k=j$. Suppose that $i<j$. Then by the definition of $D_j$,
		\[K_{r_i}^{D_j}(z)=\min\{\eta r_j,K_{r_i}^{D_{j-1}}(z)\}+O(\log r_j)\,,\]
		and by the induction hypothesis,
		\[K_{r_i}^{D_{j-1}}(z) = \eta r_1 + \sum_{l =2}^i \min\{\eta (r_l - r_{l-1}), K_{r_l, r_{l-1}}(z \, | \, z)\} + O(\log r_{j-1})\,.\]
		Since
		\[\eta r_1 + \sum_{l =2}^i \min\{\eta (r_l - r_{l-1}), K_{r_l, r_{l-1}}(z \, | \, z)\} \leq \eta r_i\,,\]
		we have 		
		\begin{equation*}
		K^{D_j}_{r_i}(z) = \eta r_1 + \sum_{l =2}^i \min\{\eta (r_l - r_{l-1}), K_{r_l, r_{l-1}}(z \, | \, z)\} + O(\log r_j)
		\end{equation*}
		Now suppose that $i = j$. If $K_{r_j}^{D_{j-1}}(z)<\eta r_j$, then, by our induction hypothesis and Lemma \ref{lem:oracles},
		\begin{align*}
		K_{r_i}^{D_j}(z)=&K_{r_i}^{D_{j-1}}(z)
		\\=& K_{r_{i-1}}^{D_{j-1}}(z) + K^{D_{j-1}}_{r_i, r_{i-1}}(z \, | \, z) - O(\log r_{j})
		\\=& \eta r_1 + \sum_{l =2}^{i-1} \min\{\eta (r_l - r_{l-1}), K_{r_l, r_{l-1}}(z \, | \, z)\} + O(\log r_{j}) \\ 
		&\; + K_{r_i, r_{i-1}}(z \, | \, z) + O(\log r_{j-1}\\
		=& \eta r_1 + \sum_{l =2}^{i} \min\{\eta (r_l - r_{l-1}), K_{r_l, r_{l-1}}(z \, | \, z)\} + O(\log r_{j}) \,.
		\end{align*}
		If instead $K_{r_i}^{D_{j-1}}(z)\geq \eta r_i$, then $K_{r_i}^{D_{j}}(z) = \eta r_i-O(\log r_i)$ by Lemma~\ref{lem:oracles}, relative to $D_{j-1}$. Since $K_{r_i}^{D_{j-1}}(z)\geq \eta r_i$ implies that $K_{r_i, r_{i-1}}(z \, | \, z) \geq \eta (r_i - r_{i-1})$, 
		\begin{center}
			$K^{D_j}_{r_i}(z) = \eta r_1 + \sum_{l =2}^i \min\{\eta (r_l - r_{l-1}), K_{r_l, r_{l-1}}(z \, | \, z)\} + O(\log r_i)$
		\end{center}				
		Therefore property 2 holds for all $1 \leq i \leq k$.
		
		To complete the proof we show that property 3 is satisfied for $k=j$. Let $t \in \N$ and $y \in \R^m$. By Lemma~\ref{lem:oracles}, relativized to $D_{j-1}$, and our induction hypothesis,
		\begin{align*}
		K^{z, D_j}_{t}(y) &= K^{z, D_{j-1}}_t(y)+O(\log r_j)\\
		&= K^{z}_t(y) + O(\log r_{j-1})+O(\log r_j)\\
		&= K^z_t(y) + O(\log r_j)\,.
		\end{align*}
		Thus, by mathematical induction, the lemma holds for all $k\in\N$.
	\end{proof}
	
	\subsection{Proof of Second Main Theorem}
	{\renewcommand{\thethm}{\ref{thm:spectraMainThm2}}
		\begin{thm}
			Let $a, b \in\R$ such that $\dim(a, b) \geq 1$. Then for every $s \in [\frac{1}{2}, 1]$ there is a point $x\in\R$ such that  $\dim(x, ax + b) \in [\frac{3}{2} + s - \frac{1}{2s}, s + 1]$.
		\end{thm}
		\addtocounter{thm}{-1}
		\begin{proof}
			
			Let $s \in [\frac{1}{2},1]$ and $y \in \R$ be random relative to $(a, b)$. That is, there is some constant $c \in \N$ such that for all $r \in \N$,
			\begin{center}
				$K^{a, b}_r(y) \geq r - c$.
			\end{center}
			Define sequence of natural numbers $\{h_j\}_{j \in \N}$ inductively as follows. Define $h_0 = 1$. For every $j > 0$, define 
			\begin{equation*}
			h_j = \min\left\{h \geq 2^{h_{j-1}}: K_{h}(a, b) \leq \left(\dim(a, b) + \frac{1}{j}\right)h\right \}.
			\end{equation*}
			Note that $h_j$ always exists. For every $r \in \N$, let
			\begin{align*}
			x[r] = \begin{cases}
			0 &\text{ if } \frac{r}{h_j} \in \left(1, \frac{1}{s}\right] \text{ for some } j \in \N \\
			y[r] &\text{ otherwise}
			\end{cases}
			\end{align*}
			Define $x \in \R$ to be the real number with this binary expansion. Then,
			\[K_{h_j}(x) = h_j + O(\log h_j)\,.\]
			
			We first show that $\dim(x, ax + b) \leq s + 1$. For every $j \in \N$, 
			\begin{align*}
			K_{h_j/s}(x,ax+b)&= K_{h_j/s}(x) + K_{h_j/s}(ax + b \, | \, x) + O(\log h_j/s) \\
			&= K_{h_j}(x) + K_{h_j/s}(ax + b \, | \, x) + O(\log h_j) \\
			&\leq h_j + 1\cdot h_j/s + o(h_j).
			\end{align*}
			Therefore, 
			\begin{align*}
			\dim(x, ax + b) &= \liminf_{r \to \infty} \frac{K_r(x, ax + b)}{r}\\
			&\leq \liminf_{j \to\infty} \frac{K_{h_j/s}(x, ax + b)}{h_j/s}\\
			&\leq \liminf_{j \to \infty} \frac{sh_j + h_j  +o(h_j)}{h_j}\\
			&= s + 1.
			\end{align*}
			
			Let $H = \Q \cap (s, 1)$, and $\eta \in H$. Let $\eta^\prime \in \Q \cap (0, s]$, $\delta = 1 - \eta > 0$, and $\ve \in \Q_+$. Let $j \in \N$. We first show that 
			\begin{equation}\label{eq:mainThm2eq}
			K_{r}(x, ax + b) \geq s r + \eta r - c\frac{\ve}{\delta} r - o(r),
			\end{equation}
			for every $r \in (h_j, 2h_j]$. Let $r \in (h_j, 2h_j]$. Let $r_1 = h_j$, $r_2 = r$, and $D_{r} = D(r_1, r_2, (a, b), \eta)$ be the oracle defined in Lemma \ref{lem:oraclesSpectra}. We first note that, by our construction of $x$, 
			\begin{align*}
			K_{r, r_1}(a, b \, | \, a, b) &= K_r(a, b) - K_r(a, b) + O(\log r)\\
			&\geq K_r(a, b) - \dim(a, b)r_1 - h_j/j + O(\log r)\\
			&\geq \dim(a, b)r - \dim(a, b)r_1 - h_j/j + O(\log r)\\
			&\geq \dim(a, b) (r - r_1) - h_j/j + O(\log r)\\
			&> \eta (r - r_1) - h_j/j + O(\log r).
			\end{align*}
			Hence, by property 2 of Lemma \ref{lem:oraclesSpectra}
			\begin{equation}\label{eq:boundOracle2}
			\eta r - h_j/j - O(\log r) \leq K^{D_r}_r(a, b) \leq \eta r + O(\log r).
			\end{equation}
			
			We now show that the conditions of Lemma~\ref{lem:pointInduct} are satisfied. By Lemma~\ref{lem:oraclesSpectra}, for each $i \in \{1, 2\}$,
			\begin{equation*}
			K^{D_r}_{r_i}(a, b) \leq \eta r_i + O(\log r_2)\,.
			\end{equation*}
			Hence, condition 2 of Lemma \ref{lem:pointInduct} is satisfied. 
			
			To see that condition 3 is satisfied for $i = 1$, let $(u, v) \in B_1(a, b)$ such that $ux + v = ax + b$ and $t=-\log\|(a,b)-(u,v)\| \leq r_1$. Then, by Lemmas \ref{lem:lines} and \ref{lem:oraclesSpectra}, and our construction of $x$,
			\begin{align*}
			K^{D_{r}}_{r_1}(u,v) &\geq K^{D_{r}}_t(a,b) + K^{D_{r}}_{r_1-t,r_1}(x|a,b)-O(\log r_1)\\
			&\geq \min\{\eta r_1, K_t(a, b)\} + K_{r_1-t}(x)- o(r_k)\\
			&\geq \min\{\eta r_1, \dim(a, b) t - o(t)\} + (\eta + \delta)(r_1 - t) - o(r_k)\\ 
			&\geq \min\{\eta r_1, \eta t - o(t)\} + (\eta + \delta)(r_1 - t) - o(r_k) \\ 
			&\geq \eta t - o(t) + (\eta + \delta)(r_1 - t) - o(r_k)\\
			&\geq (\eta - \ve)r_1 + \delta(r_1 - t)
			\end{align*}
			for all sufficiently large $j$.
			
			To see that that condition 3 is satisfied for $i = 2$, let $(u, v) \in B_{2^{-r_{1}}}(a, b)$ such that $ux + v = ax + b$ and $t=-\log\|(a,b)-(u,v)\| \leq r_2$. Since $(u, v) \in B_{2^{-r_{1}}}(a, b)$,
			\begin{align*}
			r_2 - t &\leq r_2 - r_{1}\\
			&\leq 2r_1 - r_1  \\
			&= r_1.
			\end{align*}
			Therefore, by Lemmas \ref{lem:lines} and \ref{lem:oraclesSpectra}, inequality (\ref{eq:boundOracle2}) and our construction of $x$,
			\begin{align*}
			K^{D_{r}}_{r_2}(u,v) &\geq K^{D_{r}}_t(a,b) + K^{D_{r}}_{r_2-t,r_2}(x|a,b)-O(\log r_2)\\
			&\geq \min\{\eta r_2, K_t(a, b)\} + K_{r_2-t}(x)- o(r_2)\\
			&\geq \min\{\eta r_2, \eta t - h_j/j - o(t)\} + (\eta + \delta)(r_2 - t) - o(r_2)\\ 
			&\geq \eta t - h_j/j - o(t) + (\eta + \delta)(r_2 - t) - o(r_2)\\
			&= \eta r_2 - h_j/j - o(t) + \delta(r_2 - t) - o(r_2)\\
			&\geq \eta r_2 - r_2/j - o(t) + \delta(r_2 - t) - o(r_2)\\
			&\geq (\eta - \ve)r_2 + \delta(r_2 - t),
			\end{align*}
			for all sufficiently large $j$. Hence the conditions of Lemma \ref{lem:pointInduct} are satisfied, and we have
			\begin{align*}
			K_{r}(x, ax + b) \geq&\ K^{D_{r}}_{r}(x, ax + b) - O(1)\\
			\geq&\ K^{D_{r}}_{r}(a, b, x) - 4\left(K(\ve) + K(\eta) + \frac{4\ve}{\delta} r + O(\log r)\right)\\
			=&\ K^{D_{r}}_{r}(a, b) + K^{D_{r}}_{r}(x \, | \, a, b) \\
			&- 4\left(K(\ve) + K(\eta) + \frac{4\ve}{\delta} r + O(\log r)\right)\\
			\geq&\ sr + \eta r - 4\left(K(\ve) + K(\eta) + \frac{4\ve}{\delta} r + O(\log r)\right).
			\end{align*}
			Hence, for every $r \in (h_j, 2h_j]$,
			\begin{align*}
			K_{r}(x, ax + b) &\geq sr + \eta r - \frac{\alpha\ve}{\delta} r - o(r)\\
			&\geq sr + \eta r - \frac{\alpha\ve}{\delta} r - o(r)
			\end{align*}
			where $\alpha$ is a fixed constant, not depending on $\eta$ and $\ve$.
			
			To complete the proof, it suffices to show that $K_r(x, ax + b) \geq r(\frac{3}{2} + s - \frac{1}{2s} - \ve)$, for every $r \in (2h_j, h_{j+1}]$. Let $r \in (2h_j, h_{j+1}]$. Then by Lemma \ref{lem:chain} and our construction of $x$,
			\begin{align*}
			K_r(x) &= K_{r, h_j/s}(x \, | \, x) + K_{h_j/s}(x) + O(\log r)\\
			&= r - h_j/s + h_j + O(\log r).
			\end{align*}
			The proof of Theorem \ref{thm:boundingMainThm} shows that 
			\begin{align*}
			K_r(x, ax + b) &\geq K_r(x) + \eta^\prime r - o(r)\\
			&\geq r - h_j/s + h_j + \eta^\prime r - o(r)\\
			&\geq r\left(\frac{3}{2} + s - \frac{1}{2s} - \ve\right)
			\end{align*}
			for sufficiently large $j$. 
			
			Since $\eta$, $\eta^\prime$ and $\ve$ were chosen arbitrarily, the conclusion follows.
		\end{proof}
	\end{document}